\documentclass[11p,reqno]{amsart}
\usepackage[T1]{fontenc}
\usepackage{graphicx}
\usepackage{amssymb,amsthm,amsmath,mathrsfs,bm,braket,marginnote}
\usepackage{enumerate}
\usepackage{appendix}
\usepackage{bbm}
\usepackage{eufrak}
\usepackage{mathrsfs}
\usepackage[colorlinks=true, pdfstartview=FitV, linkcolor=blue, citecolor=blue, urlcolor=blue]{hyperref}

\usepackage{pgf}
\usepackage{pgfplots}
\usepackage{extarrows}
\usepackage{tikz}
\usetikzlibrary{arrows,calc}
\usepackage{verbatim}
\usetikzlibrary{decorations.pathreplacing,decorations.pathmorphing}

\numberwithin{equation}{section}
\newtheorem{theorem}{Theorem}[section]

\newtheorem{proposition}[theorem]{Proposition}
\newtheorem{definition}[theorem]{Definition}
\newtheorem{remark}[theorem]{Remark}
\newtheorem{corollary}[theorem]{Corollary}

\newcommand{\p}{\partial}
\newcommand{\Pf}{\text{Pf}}

 \makeatletter
 \newcommand{\Rmnum}[1]{\expandafter\@slowromancap\romannumeral #1@}
 \makeatother
 \reversemarginpar
\begin{document}

\title[Christoffel transformations for SOPs and PSOPs]{Christoffel transformations for (partial-)skew-orthogonal polynomials and applications}

\author{Shi-Hao Li}
\address{}
\email{lishihao@lsec.cc.ac.cn}
\author{Guo-Fu Yu}
\address{School of Mathematical Sciences, Shanghai Jiaotong University, People's Republic of China.}
\email{gfyu@sjtu.edu.cn}

\dedicatory{In celebration of Professor Peter J Forrester's $60^{\text{th}}$ birthday}

\date{}


\begin{abstract}
In this article, we consider the Christoffel transformations for skew-orthogonal polynomials and partial-skew-orthogonal polynomials. We demonstrate that the Christoffel transformations can act as spectral problems for discrete integrable hierarchies, and therefore we derive certain integrable hierarchies from these transformations. Some reductional cases are also considered.
\end{abstract}

\maketitle
\section{Introduction}

The theory of orthogonal polynomials is an important topic in modern analysis. In particular, it has many different applications in contexts of mathematical physics such as random matrix theory and integrable system, for example, see \cite{deift99,forrester10,hjn16,mehta04,vanassche17}. One of the key features of an orthogonal polynomial sequence is the three term recurrence. This relation, also referred to as the spectral problem, is connected with integrable systems when time evolutions are permitted. Apart from the continuous spectral problem, attention is also paid to the discrete spectral transformation, namely the Christoffel transformation.
The Christoffel transformation is given in terms of an adjacent family of orthogonal polynomials \cite{spicer06,tsujimoto00}
\begin{align*}
P_n^{(m)}(x)=\frac{1}{\tau_{n-1}^{(m)}}\det\left(\begin{array}{ccc}
c_m&\cdots&c_{n+m}\\
\vdots&&\vdots\\
c_{n+m-1}&\cdots&c_{2n+m-1}\\
1&\cdots&x^n\end{array}
\right),\quad \tau_{n-1}^{(m)}=\det(c_{m+i+j})_{i,j=0,\cdots,n-1}.
\end{align*}
By directly using determinant identities (see \cite[Eqs. (1.1.22a-b)]{spicer06}), one can find the recurrence
\begin{align*}
P_n^{(m)}(x)&=xP_{n-1}^{(m+1)}(x)-\frac{\tau_{n-1}^{(m+1)}\tau_{n-2}^{(m)}}{\tau_{n-2}^{(m+1)}\tau_{n-1}^{(m)}}P_{n-1}^{(m)}(x),\\ P_n^{(m)}(x)&=xP_{n-1}^{(m+2)}(x)-\frac{\tau_{n-1}^{(m+1)}\tau_{n-2}^{(m+1)}}{\tau_{n-1}^{(m)}\tau_{n-2}^{(m+2)}}P_{n-1}^{(m+1)}(x).
\end{align*}
This kind of spectral transformation is useful in finding a relationship with discrete integrable systems since the discrete index $m$ naturally appears in the adjacent orthogonal polynomials. In the literature, there are many applications of Christoffel transformations to classical integrable systems. For example, in \cite{chang14,spicer11}, the Christoffel transformations were applied to finding discrete Toda systems with higher analogues, which are related to the qd and qqd algorithms. In \cite{alv17}, Christoffel transformation for matrix orthogonal polynomials was considered, and its connection with non-abelian 2D Toda lattice hierarchy was found. Moreover, the Christoffel transformation for multivariate orthogonal polynomials was considered in \cite{ari18}, and its connection with integrable system was considered in \cite{ari14}.

In this work, we mainly consider the Christoffel transformations for skew orthogonal polynomials (SOPs) and partial-skew-orthogonal polynomials (PSOPs) with applications in the theory of classical integrable systems. SOPs are well known in the studies of random matrix theory, as they are the characteristic polynomials of celebrated orthogonal and symplectic ensembles with their specified Pfaffian structures.
Besides, these polynomials are also applicable to integrable system. A connection between SOPs and the so-called Pfaff lattice was firstly considered in \cite{adler99,adler02}, and later considered in the geometric setting \cite{kodama07,kodama09}. In \cite{miki12}, the discrete Pfaff lattice was considered by using the discrete spectral transformation of SOPs with Pfaffian tau functions.

We emphasise that Pfaffian tau functions are not only important in integrable systems \cite{adler022} but play a significant role in quantum field theory such as 2D Ising model, dimer models and 1D XY chain \cite{auyang87,perk84}. Therefore, Pfaffian tau functions are very worthy of study.
We remark that in addition to the above mentioned SOPs and even-order Pfaffian tau functions, one can obtain odd-order Pfaffian tau functions from a generalised Wick's theorem.
One can therefore naturally ask about the odd-order Pfaffian tau functions and corresponding polynomials theory. In \cite{chang182}, the concept of PSOPs was proposed and the reason to call these polynomials PSOPs is that these odd-order polynomials are not skew orthogonal with the even ones. Though not skew-orthogonal, by making use of these polynomials, many interesting integrable lattice were found with applications in convergence acceleration algorithms, vector Pad\'e approximation and condensation algorithms for Pfaffians \cite{li20}. More importantly, one specified PSOPs is related to the Bures ensemble with potential application in quantum information theory \cite{forrester16}.
Therefore, the Christoffel transformations for SOPs and PSOPs are not merely important in orthogonal polynomials theory itself but with potential and established applications in many other subjects.

In Section \ref{ct}, we firstly give a brief review of SOPs and PSOPs.
By employing Pfaffian identities, we give their Christoffel transformations. For SOPs $\{P_n^{(m)}(z)\}_{n,m\in\mathbb{N}}$, we have the transformations
\begin{align*}
&P_{2n+1}^{(m)}(z)-\mathcal{A}^m_n P_{2n}^{(m)}(z)=z\left(
P_{2n}^{(m+1)}(z)-\mathcal{B}^m_n P_{2n-2}^{(m+1)}(z)
\right),\\
&P_{2n+2}^{(m)}(z)-\mathcal{C}^m_n P_{2n}^{(m)}(z)=z\left(
P_{2n+1}^{(m+1)}(z)-\mathcal{D}^m_nP_{2n}^{(m+1)}(z)\right),
\end{align*}
with proper coefficients $\mathcal{A}_n^m$, $\mathcal{B}_n^m$, $\mathcal{C}_n^m$ and $\mathcal{D}_n^m$. This Christoffel transformation is slightly different from the one in \cite[Thm. 3]{miki12} since here the formula only involves two adjacent families of SOPs. Regarding PSOPs $\{Q_n^{(m)}(z)\}_{n,m\in\mathbb{N}}$, the Christoffel transformation could be  identically written as
\begin{align*}
Q^{(m)}_{n+1}(z)+\mathcal{\xi}^{m}_nQ^{(m)}_n(z)=z\left( Q_n^{(m+1)}(z)+\mathcal{\eta}^m_nQ^{(m+1)}_{n-1}(z) \right),
\end{align*}
with coefficients $\xi_n^m$ and $\eta_n^m$ properly chosen. Moreover, we find a multi-component version of odd-order PSOPs, and therefore give a multi-component Christoffel transformation as well.

In Section \ref{sop_ct}, we manifest how to make use of the Christoffel transformation of SOPs. By expanding the SOPs in terms of monomial with coefficients expressed by Pfaffian tau functions, one can easily get the DKP (or Pfaff-lattice) hierarchy from the Christoffel transformation. Moreover, we consider a reductional case---Laurent type SOPs \cite{miki20}, by which the Christoffel transformation is reduced to a three term recurrence relation and the corresponding integrable hierarchy is reduced to the 1d-Toda hierarchy with wave function expressed as SOPs.

In Section \ref{psop_ct}, some considerations are taken into the Christoffel transformation of PSOPs. The most general case is firstly given and then some reductional cases are considered. We demonstrate three different examples to show how to make use of moment constraint approach to obtain lower dimensional integrable lattices, generalising the moment constraint approach proposed in \cite{li19}. Concluding remarks are given in Section \ref{con}.

\section{Christoffel transformations of skew orthogonal polynomials and partial skew orthogonal polynomials}\label{ct}
The main purpose of this part is to derive the Christoffel transformation for the skew-orthogonal polynomials (SOPs) and partial-skew-orthogonal polynomials (PSOPs) as an analogy of the that for orthogonal polynomials.
Such transformations can be regarded as the spectral transformations, thus being prepared for the later discussion about how to connect with integrable systems after the involvement of time. To this end, we firstly need to give introductions to SOPs and PSOPs. Following \cite{chang182}, we start with a skew symmetric inner product, and then give some brief derivations about SOPs and PSOPs from a unified framework.

\subsection{Skew symmetric inner product, SOPs and PSOPs}
Let us consider a skew symmetric inner product $\langle\cdot,\cdot\rangle$ from $\mathbb{R}[z]\times\mathbb{R}[z]\to \mathbb{R}$ satisfying the skew symmetry property $$\langle f(z),g(z)\rangle=-\langle g(z),f(z)\rangle,$$ and define the skew symmetric bi-moments
\begin{align*}
\mu_{i,j}=\langle z^i,z^j\rangle=-\langle z^j,z^i\rangle=-\mu_{j,i}.
\end{align*}
Then we investigate the (skew-)orthogonality under the skew symmetric inner product.

For monic polynomials $\{P_n(z)\}_{n\in\mathbb{N}}$, if we consider the orthogonal conditions
\begin{align*}
\langle P_n(z),z^i\rangle=0,\quad 0\leq i\leq n-1,
\end{align*}
then as that discussed in \cite[Sec. 2]{chang182}, only the even family of polynomials are well-defined and the odd ones are not.
Therefore, how to set up the inner product condition to make the odd-order polynomials well defined is a key point at this stage. One suggestive way is to consider the conditions
\begin{align*}
&\langle P_{2n}(z), z^i\rangle=0,\qquad\qquad\quad\,\,\,0\leq i\leq 2n,\\
&\langle P_{2n+1}(z),z^i\rangle=\alpha_{2n+1,i},\qquad 0\leq i\leq 2n+1
\end{align*}
where $\{\alpha_{2n+1,i}\}_{i=0}^{2n+1}$ are $2n+2$ parameters satisfying
\begin{align}\label{det}
\det\left(\begin{array}{cccc}
\mu_{0,0}&\cdots& \mu_{2n,0}&\mu_{2n+1,0}-\alpha_{2n+1,0}\\
\mu_{0,1}&\cdots&\mu_{2n,1}&\mu_{2n+1,1}-\alpha_{2n+1,1}\\
\vdots&&\vdots&\vdots\\
\mu_{0,2n+1}&\cdots&\mu_{2n,2n+1}&\mu_{2n+1,2n+1}-\alpha_{2n+1,2n+1}
\end{array}
\right)=0.
\end{align}
By differently choosing $\{\alpha_{2n+1,i}\}_{i=0}^{2n+1}$, we get different families of odd-order polynomials.
\subsubsection{Skew-orthogonal polynomials $\{P_n(z)\}_{n\in\mathbb{N}}$}
The choices
\begin{align*}
\alpha_{2n+1,i}=-\frac{\tau_{2n+2}}{\tau_{2n}}\delta_{i,2n},\quad 0\leq i\leq 2n+1
\end{align*}
give rise to the concept of SOPs, where $\tau_{2n}=\Pf(0,\cdots,2n-1)$ and $\Pf(i,j)=\mu_{i,j}$.
Therefore, one can get SOPs $\{P_n(z)\}_{n\in\mathbb{N}}$ by requiring the skew orthogonal relations
\begin{align}\label{sop}
 \left\langle P_{2n}(z),P_{2m}(z)\right\rangle=\left\langle P_{2n+1}(z),P_{2m+1}(z)\right\rangle=0,\quad  \left\langle P_{2m}(z),P_{2n+1}(z)\right\rangle=\frac{\tau_{2n+2}}{\tau_{2n}}\delta_{m,n}.
\end{align}
The condition \eqref{sop} is indeed a consistent linear system for the coefficients of polynomials. By solving it and applying a Jacobi identity, one can find Pfaffian expressions for SOPs \cite{adler99,chang182}
\begin{align*}
P_{2n}(z)=\frac{1}{\tau_{2n}}\Pf(0,\cdots,2n,z),\quad P_{2n+1}(z)=\frac{1}{\tau_{2n}}\Pf(0,\cdots,2n-1,2n+1,z)
\end{align*}
with $\Pf(i,z)=z^i$.
\subsubsection{Partial-skew-orthogonal polynomials $\{Q_n(z)\}_{n\in\mathbb{N}}$} Except the choice demonstrated above, there is another choice to introduce $2n+2$ quantities $\{\beta_i\}_{i=0}^{2n+1}$ such that
 \begin{align*}
 \alpha_{2n+1,i}=-\beta_i\frac{\tau_{2n+2}}{\tau_{2n+1}},\quad \tau_{2n+1}=\Pf(d,0,\cdots,2n)
 \end{align*}
 with $\Pf(d,i)=\beta_i$ and $\Pf(i,j)=\mu_{i,j}$. Here the quantities $\{\beta_i\}_{i=0}^{2n+1}$ are chosen so that $\tau_{2n+1}\ne0$. Verifications of the condition \eqref{det} is based on a Jacobi identity (see \cite{chang182} for more details).
In this case, the skew orthogonal relation can be formulated as follows
\begin{align*}
\langle Q_{2n}(z), z^i\rangle=\frac{\tau_{2n+2}}{\tau_{2n}}\delta_{2n+1,i},\quad \langle Q_{2n+1}(z). z^i\rangle=-\beta_i\frac{\tau_{2n+2}}{\tau_{2n+1}},\quad {0\leq i\leq 2n+1}.
\end{align*}
Moreover, these relations admit the following Pfaffian expressions
\begin{align*}
Q_{2n}(z)=\frac{1}{\tau_{2n}}\Pf(0,\cdots,2n,z),\quad Q_{2n+1}(z)=\frac{1}{\tau_{2n+1}}\Pf(d,0,\cdots,2n+1,z)
\end{align*}
with $\Pf(d,z)=0$. It is remarkable that both even- and odd-order PSOPs are uniquely determined, although the odd ones can be arbitrarily chosen due to the freedom of $\{\beta_j\}_{j=0}^{2n+1}$.
Therefore, by assuming that there are $\ell$ different sets $\{\beta_j^{(k)}\}_{j=0}^{2n+1}$ for $k=1,\cdots,\ell$ such that for each $k$, $\tau_{2n+1,k}=\Pf(d_k,0,\cdots,2n+1)\ne0$ with $\Pf(d_k,i)=\beta_i^{(k)}$, we can define $\ell$-component PSOPs of odd order satisfying the relations
\begin{align*}
Q_{2n+1,k}(z)=\frac{1}{\tau_{2n+1,k}}\Pf(d_k,0,\cdots,2n+1,z),\quad \langle Q_{2n+1,k}(z), z^i\rangle=-\beta_i^{(k)}\frac{\tau_{2n+2}}{\tau_{2n+1,k}},\quad k=1,\cdots,\ell.
\end{align*}

\subsection{Christoffel transformations of SOPs}
It is important to develop the Christoffel transformations of orthogonal polynomials since such transformations can act as the spectral problem and characterise the property of polynomials. Some previous results about the discrete spectral transformations of SOPs are based on the evolution of the functional \cite{miki12}, but we emphasise on the adjacent families of SOPs and consider corresponding Christoffel transformations. It is remarkable that the adjacent SOPs here are the special $\mu=0$ case in \cite[Thm. 2]{miki12}, but the Christoffel transformation is only between two different families of polynomials and different from the known results.
\begin{definition}
For $m\in\mathbb{N}$, the $m$-th adjacent family of SOPs is defined by
\begin{align*}
P_{2n}^{(m)}=\frac{1}{z^m\tau_{2n}^{(m)}}\Pf(m,\cdots,m+2n,z),\quad P_{2n+1}^{(m)}=\frac{1}{z^m\tau_{2n}^{(m)}}\Pf(m,\cdots,m+2n-1,m+2n+1,z),
\end{align*}
where $\tau_{2n}^{(m)}=\Pf(m,\cdots,m+2n-1)$.
\end{definition}
One of the most significant features of the adjacent family of polynomials is that they inherit the skew-orthogonality under the modified inner product $\langle z^m\cdot,z^m\cdot\rangle$. To be precise, we have
\begin{align*}
 &\left\langle z^mP_{2n}^{(m)}(z),z^mP_{2l}^{(m)}(z)\right\rangle=\left\langle z^mP_{2n+1}^{(m)}(z),z^mP_{2l+1}^{(m)}(z)\right\rangle=0,\\
 & \left\langle z^mP^{(m)}_{2l}(z),z^mP^{(m)}_{2n+1}(z)\right\rangle=\frac{\tau^{(m)}_{2n+2}}{\tau^{(m)}_{2n}}\delta_{l,m}.
\end{align*}
The existence and uniqueness of the adjacent family of SOPs are equivalent to the condition $\tau_{2n}^{(m)}\ne0$, and there are many physically interesting examples such as partition functions of orthogonal/symplectic ensembles satisfying such a condition.
\begin{proposition}\label{SOP-ct}
The Christoffel transforms for SOPs has the form
\begin{subequations}
\begin{align}
&P_{2n+1}^{(m)}(z)-\mathcal{A}^m_n P_{2n}^{(m)}(z)=z\left(
P_{2n}^{(m+1)}(z)-\mathcal{B}^m_n P_{2n-2}^{(m+1)}(z)
\right),\label{sop-ct1}\\
&P_{2n+2}^{(m)}(z)-\mathcal{C}^m_n P_{2n}^{(m)}(z)=z\left(
P_{2n+1}^{(m+1)}(z)-\mathcal{D}^m_nP_{2n}^{(m+1)}(z)\right),\label{sop-ct2}
\end{align}
\end{subequations}
with coefficients
\begin{align*}
\mathcal{A}^m_n=\frac{P_{2n+1}^{(m)}(0)}{P_{2n}^{(m)}(0)},\quad
\mathcal{B}^m_n=\frac{\tau_{2n+2}^{(m)}\tau_{2n-2}^{(m+1)}}{\tau_{2n}^{(m)}\tau_{2n}^{(m+1)}},\quad
\mathcal{C}^m_n=\frac{\tau_{2n}^{(m)}\tau_{2n+2}^{(m+1)}}{\tau_{2n+2}^{(m)}\tau_{2n}^{(m+1)}},\quad
\mathcal{D}^m_n=\frac{P_{2n+3}^{(m-1)}(0)}{P_{2n+2}^{(m-1)}(0)}.
\end{align*}
\end{proposition}
\begin{proof}
Starting from the Pfaffian identity
\begin{align*}
&\Pf(m,\ast)\Pf(\ast,2n+m,2n+m+1,z)=\Pf(\ast,2n+m)\Pf(m,\ast,2n+m+1,z)\\
&\qquad\qquad-\Pf(\ast,2n+m+1)\Pf(m,\ast,2n+m,z)+\Pf(\ast,z)\Pf(m,\ast,2n+m,2n+m+1)
\end{align*}
with $\{\ast\}=\{m+1,\cdots,2n+m-1\}$
and recognising the fact that
\begin{align}\label{co}
P_{2n}^{(m)}(0)=\frac{\tau_{2n}^{(m+1)}}{\tau_{2n}^{(m)}},\quad P_{2n+1}^{(m)}(0)=\frac{1}{\tau_{2n}^{(m)}}\Pf(m+1,\cdots,2n+m-1,2n+m+1)
\end{align}
we  get the identity \eqref{sop-ct1}.

The identity \eqref{sop-ct2} can be obtained from the Pfaffian identity
\begin{align*}
&\Pf(m,\ast,2n+m+1,2n+m+2,z)\Pf(\ast)=\Pf(m,\ast,2n+m+1)\Pf(\ast,2n+m+2,z)\\
&\qquad-\Pf(m,\ast,2n+m+2)\Pf(\ast,2n+m+1,z)+\Pf(m,\ast,z)\Pf(\ast,2n+m+1,2n+m+2)
\end{align*}
with $\{\ast\}=\{m+1,\cdots,2n+m\}$, and the term $\Pf(m,\cdots,2n+m,2n+m+2)$ can be written in terms of SOPs with the help of \eqref{co}.
\end{proof}

\subsection{Christoffel transformations of PSOPs}
Similar to the adjacent family of SOPs, we now consider the adjacent family of PSOPs. Here we focus on the multi-component case for odd-order polynomials since the one-component $m=1$ case was implicitly given in \cite[Sec. 3.1]{chang182}.
\begin{definition}
The $m$-th adjacent family of PSOPs are defined by
\begin{align*}
Q_{2n}^{(m)}(z)=\frac{1}{z^m \tau_{2n}^{(m)}}\Pf(m,\cdots,m+2n,z),\quad Q_{2n+1,k}^{(m)}(z)=\frac{1}{z^{m}\tau_{2n+1,k}^{(m)}}\Pf(d_k,m,\cdots,m+2n+1,z),
\end{align*}
where $\tau_{2n}^{(m)}=\Pf(m,\cdots,m+2n-1)$ and $\tau_{2n+1,k}^{(m)}=\Pf(d_k,m,\cdots,m+2n)$.
\end{definition}

The skew inner product properties of the adjacent families are easily obtained as follows
\begin{align*}
\langle z^mQ_{2n}^{(m)}(z), z^{m+i}\rangle=\frac{\tau_{2n+2}^{(m)}}{\tau_{2n}^{(m)}}\delta_{2n,i-1},\quad \langle z^m Q_{2n+1,k}^{(m)}(z),z^{m+i}\rangle=-\beta^{(k)}_{m+i}\frac{\tau_{2n+2}^{(m)}}{\tau_{2n+1,k}^{(m)}},
\end{align*}
and the existence and uniqueness of the adjacent PSOPs are equivalent to the facts that $\tau_{2n}^{(m)}\ne0$ and $\tau_{2n+1,k}^{(m)}\ne0$.

\begin{theorem}
The Christoffel transformation for one-component PSOPs is given by
\begin{align}\label{psop-ct}
& Q^{(m)}_{n+1}(z)+\mathcal{\xi}^{m}_nQ^{(m)}_n(z)=z\left( Q_n^{(m+1)}(z)+\mathcal{\eta}^m_nQ^{(m+1)}_{n-1}(z) \right)
\end{align}
with coefficients
\begin{align*}
\xi^m_n=\frac{\tau_{n}^{(m)}\tau_{n+1}^{(m+1)}}{\tau_{n+1}^{(m)}\tau_n^{(m+1)}},\quad
\eta_n^m=\frac{\tau_{n+2}^{(m)}\tau_{n-1}^{(m+1)}}{\tau_{n+1}^{(m)}\tau_{n}^{(m+1)}}.
\end{align*}
\end{theorem}
\begin{proof}
In fact, identity \eqref{psop-ct} is composed of two different situations. When $n$ is even, we take the index set $\{\ast\}=\{m+1,\cdots,2n+m\}$ and
make use of the Pfaffian identity
\begin{align*}
&\Pf(d,m,\ast,2n+m+1,z)\Pf(\ast)=\Pf(d,m,\ast)\Pf(\ast,2n+m+1,z)\\
&\qquad-\Pf(d,\ast,2n+m+1)\Pf(m,\ast,z)+\Pf(d,\ast,z)\Pf(m,\ast,2n+m+1),
\end{align*}
then we  get
\begin{align*}
Q_{2n+1}^{(m)}(z)+\frac{\tau_{2n}^{(m)}\tau_{2n+1}^{(m+1)}}{\tau_{2n+1}^{(m)}\tau_{2n}^{(m+1)}}Q_{2n}^{(m)}(z)=z
\left(
Q_{2n}^{(m+1)}(z)+\frac{\tau_{2n+2}^{(m)}\tau_{2n-1}^{(m+1)}}{\tau_{2n+1}^{(m)}\tau_{2n}^{(m+1)}}Q_{2n-1}^{(m+1)}(z)
\right).
\end{align*}
If $n$ is odd, we need to shift the index $\{\ast\}=\{m+1,\cdots,m+2n+1\}$, and use the Pfaffian identity
\begin{align*}
&\Pf(d,m,\ast,m+2n+2)\Pf(\ast,z)=\Pf(d,\ast)\Pf(m,\ast,m+2n+2,z)\\
&\qquad-\Pf(m,\ast)\Pf(d,\ast,m+2n+2,z)+\Pf(\ast,m+2n+2)\Pf(d,m,\ast,z)
\end{align*}
to  get
\begin{align*}
Q_{2n+2}^{(m)}(z)+\frac{\tau_{2n+1}^{(m)}\tau_{2n+2}^{(m+1)}}{\tau_{2n+2}^{(m)}\tau_{2n+1}^{(m+1)}}Q_{2n+1}^{(m)}(z)=z\left(
Q_{2n+1}^{(m+1)}(z)+\frac{\tau_{2n+3}^{(m)}\tau_{2n}^{(m+1)}}{\tau_{2n+2}^{(m)}\tau_{2n+1}^{(m+1)}}Q_{2n}^{(m+1)}(z)
\right).
\end{align*}
Combining these results we obtain \eqref{psop-ct}.
\end{proof}

The procedure stated above implies that the Christoffel transformation for PSOPs can be extended to the multi-component case, i.e
\begin{subequations}
\begin{align}
Q_{2n+1,k}^{(m)}(z)+\mathcal{E}_{n,k}^m Q_{2n}^{(m)}(z)&=z\left(
Q_{2n}^{(m+1)}(z)+\mathcal{F}_{n,k}^m Q_{2n-1,k}^{(m+1)}(z)
\right),\label{psop-ct1}\\
Q_{2n+2}^{(m)}(z)+\mathcal{G}_{n,k}^{m} Q_{2n+1,k}^{(m)}(z)&=z\left(
Q_{2n+1,k}^{(m+1)}(z)+\mathcal{H}_{n,k}^m Q_{2n}^{(m+1)}(z)
\right),\label{psop-ct2}
\end{align}
\end{subequations}
where the coefficients $\mathcal{E}_{n,k}^m$, $\mathcal{F}_{n,k}^m$, $\mathcal{G}_{n,k}^m$, $\mathcal{H}_{n,k}^m$ are given by
\begin{align*}
\mathcal{E}_{n,k}^m=\frac{\tau_{2n}^{(m)}\tau_{2n+1,k}^{(m+1)}}{\tau_{2n+1,k}^{(m)}\tau_{2n}^{(m+1)}},\quad \mathcal{F}_{n,k}^m=\frac{\tau_{2n+2}^{(m)}\tau_{2n-1,k}^{(m+1)}}{\tau_{2n+1,k}^{(m)}\tau_{2n}^{(m+1)}},\quad \mathcal{G}_{n,k}^m=\frac{\tau_{2n+1,k}^{(m)}\tau_{2n+2}^{(m+1)}}{\tau_{2n+2}^{(m)}\tau_{2n+1,k}^{(m+1)}},\quad \mathcal{H}_{n,k}^m=\frac{\tau_{2n+3,k}^{(m)}\tau_{2n}^{(m+1)}}{\tau_{2n+2}^{(m)}\tau_{2n+1,k}^{(m+1)}}.
\end{align*}

\section{Applications of SOPs' Christoffel transformation}\label{sop_ct}

In this part, we introduce the commuting time flows, and make use of the SOPs' Christoffel transformation to obtain integrable lattices. The concept of commuting flows were proposed in \cite{adler99} by considering the evolutions of moment matrices $\mathcal{U}:=(\mu_{i,j})_{i,j\in\mathbb{N}}$ such that $\p_{t_n}\mathcal{U}=\Lambda^n \mathcal{U}+\mathcal{U}\Lambda^{\top n}$, where $\Lambda$ is the shift operator whose off-diagonals are $1$ and the others are $0$. Such evolutions hold valid for each bi-moment, so we have
\begin{align}\label{te}
\p_{t_n}\mu_{i,j}=\mu_{i+n,j}+\mu_{i,j+n}.
\end{align}
One of the most important property under the commuting flow is to find explicitly  the derivative relationship between $P^{(m)}_{2n}(z)$ and $P_{2n+1}^{(m)}(z)$.
\begin{proposition}\label{prop1}
With time evolution \eqref{te}, it holds that
\begin{align}\label{derivative}
(z+\p_{t_1})(\tau_{2n}^{(m)}P_{2n}^{(m)})=\tau_{2n}^{(m)}P_{2n+1}^{(m)}(z).
\end{align}
\end{proposition}
\begin{proof}
Noting that
\begin{align*}
(z+\partial_{t_1})(\tau_{2n}^{(m)}P_{2n}^{(m)})
=(z+\partial_{t_1})z^{-m}\Pf(m,m+1,\cdots,m+2n,z),
\end{align*}
and expanding the right hand side in terms of $z$, one can find
\begin{align*}
&z^{2n+1}\Pf(m,\cdots,m+2n-1)-\sum_{k=0}^{2n-1}(-z)^{k+1}\Pf(m,\cdots,\widehat{m+k},\cdots,m+2n)\\
&\quad+\sum_{k=0}^{2n}(-z)^k
\left[\Pf(m,\cdots,\widehat{m+k-1},\cdots,m+2n)+\Pf(m,\cdots,\widehat{m+k},\cdots,\widehat{m+2n},m+2n+1)\right].
\end{align*}
Eliminating the last term in the first line and the first term in the last implies  \eqref{derivative}.
\end{proof}
The proof of the case $m=0$ case was given in \cite[Lemma 3.6]{adler99}. However, Proposition \ref{prop1} shows that this property  holds for all adjacent families of SOPs. Moreover,  the coefficients of SOPs can be written in terms of Schur polynomials acting on the normalisation factor (i.e. tau function) \cite[Sec. 3]{adler02}
\begin{align}\label{schur}
P_{2n}^{(m)}(z)=\frac{1}{\tau_{2n}^{(m)}}\sum_{k=0}^{2n}z^{2n-k}s_k(-\tilde{\p}_t)\tau_{2n}^{(m)},
\end{align}
where $\{s_k(t)\}_{k\in\mathbb{N}}$ are the Schur polynomials given by
\begin{align*}
\exp\left(
\sum_{\ell=1}^\infty t_\ell z^\ell
\right)=\sum_{k=0}^\infty s_k(t)z^k,
\end{align*}
and $\tilde{\partial}_t=\left(\p_{t_1},\p_{t_2}/2,\p_{t_3}/3,\cdots\right)$. Substituting expressions   \eqref{derivative} and \eqref{schur} into \eqref{sop-ct1}-\eqref{sop-ct2} and recognising
\begin{align*}
\mathcal{A}_n^m=\p_{t_1}\log\tau_{2n}^{(m+1)},\quad
\mathcal{B}^m_n=\frac{\tau_{2n+2}^{(m)}\tau_{2n-2}^{(m+1)}}{\tau_{2n}^{(m)}\tau_{2n}^{(m+1)}},\quad
\mathcal{C}^m_n=\frac{\tau_{2n}^{(m)}\tau_{2n+2}^{(m+1)}}{\tau_{2n+2}^{(m)}\tau_{2n}^{(m+1)}},\quad
 \mathcal{D}_n^m=\p_{t_1}\log\tau_{2n+2}^{(m)},
\end{align*}
by comparing with the coefficients of monomials, one can immediately get the following bilinear identities
\begin{align}
\begin{aligned}\label{dkp}
\tau_{2n}^{(m+1)}&s_{2n+1-\ell}(-\tilde{\p}_t)\tau_{2n}^{(m)}+\tau_{2n}^{(m+1)}\p_{t_1}s_{2n-\ell}(-\tilde{\p}_t)\tau_{2n}^{(m)}-\p_{t_1}\tau_{2n}^{(m+1)}s_{2n-\ell}(-\tilde{\p}_t)\tau_{2n}^{(m)}\\
&\qquad=\tau_{2n}^{(m)}s_{2n+1-\ell}(-\tilde{\p}_t)\tau_{2n}^{(m+1)}-\tau_{2n+2}^{(m)}s_{2n-1-\ell}(-\tilde{\p}_t)\tau_{2n-2}^{(m+1)}.
\end{aligned}
\end{align}
The first nontrivial case is the case of $\ell=2n-1$. In this case, \eqref{dkp} has the form\footnote{The operator $D_t$ is usually called as the Hirota's bilinear operator, defined by \begin{align*}
D_t f\cdot g=
\frac{\p}{\p s}f(t+s)g(t-s)|_{s=0}.
\end{align*}}
\begin{align*}
(D_{t_2}+D_{t_1}^2)\tau_{2n}^{(m+1)}\cdot\tau_{2n}^{(m)}=2\tau_{2n+2}^{(m)}\tau_{2n-2}^{(m+1)}.
\end{align*}
Regarding the classification results of Kyoto School, this is exactly the KP equation of $D_\infty$ type \cite{jimbo83}, and later   recognised as Pfaff lattice hierarchy \cite{adler99,adler022,adler02}.
\begin{remark}
A natural integrable discretisation of the $t_1$-flow has been considered in \cite{miki12}
\begin{align*}
\mu_{i,j}^{\ell+1}=\mu_{i+1,j+1}^\ell+\lambda\mu_{i,j+1}^\ell+\lambda\mu_{i+1,j}^\ell+\lambda^2\mu_{i,j}^\ell,
\end{align*}
where $\ell$ is a discrete index. Combining the Christoffel transformation \eqref{sop-ct1}-\eqref{sop-ct2} and the discrete evolution gives rise to the fully discrete DKP equation.
\end{remark}

\subsection{Geronimus transformation, Laurent type SOPs, and Toda lattice}\label{stoda}
Despite the above discussed Christoffel transformation, Geronimus transformation is another important discrete transformation in the orthogonal polynomials theory. The aim of the Geronimus transformation is to express the adjacent family of polynomials in terms of the original ones, namely
\begin{align}\label{ger1}
P_{n}^{(m)}(z)=\sum_{i=0}^{n}\alpha_i^{(m)} P_i^{(m+1)}(z).
\end{align}
Since these polynomials are monic. we naturally have $\alpha_n^{(m)}=1$.
Usually, the essential idea to find the Geronimus transformation is to consider the relation $\langle zP_n^{(m)}(z), z^i\rangle$ by utilising orthogonality. However, we can not get enough information to express the Geronimus transformation when  the inner product is skew symmetric only. A possible case in which we can find the Geronimus transformation is the Laurent type SOPs proposed recently in \cite{miki20}. It requires the moments
\begin{align}\label{lsop}
\mu_{i,j}=\mu_{i-1,j-1}\quad \text{or \quad  $\mathcal{U}=\Lambda U\Lambda^\top$}.
\end{align}
This condition is equivalent to the identity $\langle z^i,z^j\rangle=\langle z^{i-1},z^{j-1}\rangle$,
and thus one can prove the following proposition.
\begin{proposition}
Under the assumption \eqref{lsop}, the following holds
\begin{align}\label{ad-lsop}
P_n^{(m)}(z)=P_n^{(m+1)}(z).
\end{align}
\end{proposition}
\begin{proof}
We prove the identity $P_{2n}^{(m)}(z)=P_{2n}^{(m+1)}(z)$, while the odd-order case can be established similarly.
With use of the assumption \eqref{ger1}, we can explicitly express
\begin{align*}
 zP_{2n}^{(m)}(z)&=\sum_{i=0}^{n-1}\left(\frac{\alpha_{i}^{(m)}}{z^m}\Pf(m+1,\cdots,m+2i+1,z)
+\frac{\beta_i^{(m)}}{z^m}\Pf(m+1,\cdots,m+2i,m+2i+2,z)\right)\notag\\
&\quad+\frac{1}{z^m}\Pf(m+1,\cdots,m+2n+1,z).
\end{align*}
Taking the skew symmetric inner product on both sides with $z^j$ and note that
\begin{align*}
\langle zP_{2n}^{(m)}(z),z^j\rangle=\langle P_{2n}^{(m)}(z),z^{j-1}\rangle=0,\quad \text{ if $j-1=0,\cdots,2n$},
\end{align*}
 we conclude that $\alpha_i^{(m)}=\beta_i^{(m)}=0$ for $i=0,\cdots,n-1$.
\end{proof}
\begin{remark}
This result can also be obtained directly from \eqref{schur} by substituting $\tau_{2n}^{(m)}=\tau_{2n}^{(m+1)}$. The fact that $\tau_{2n}^{(m)}=\tau_{2n}^{(m+1)}$ could be verified by expanding the Pfaffians with use of \eqref{lsop}.
\end{remark}

The substitution of  \eqref{ad-lsop} into Christoffel transformations \eqref{sop-ct1}-\eqref{sop-ct2} yields the following identities (see \cite[Prop. 2]{miki20})
\begin{subequations}
\begin{align}
&P_{2n+1}(z)-\mathcal{A}_n P_{2n}(z)=z\left(
P_{2n}(z)-\mathcal{B}_n P_{2n-2}(z)
\right),\label{toda1}\\
&P_{2n+2}(z)-P_{2n}(z)=z\left(
P_{2n+1}(z)-\mathcal{A}_{n+1}P_{2n}(z)\right)\label{toda2}
\end{align}
\end{subequations}
where
\begin{align}\label{coe1}
\mathcal{A}_n=\p_{t_1}\log \tau_{2n},\quad \mathcal{B}_n=\frac{\tau_{2n-2}\tau_{2n+2}}{\tau_{2n}^2}.
\end{align}
Moreover,  comparing the coefficients of these polynomials, one can get the reduction of
\eqref{dkp}
\begin{align}\label{1dtoda}
D_{t_1}\tau_{2n}\cdot s_{2n-\ell}(-\tilde{\p}_t)\tau_{2n}=\tau_{2n+2}s_{2n-1-\ell}(-\tilde{\p}_t)\tau_{2n-2}.
\end{align}
The first nontrivial example of \eqref{1dtoda} is the following one
\begin{align*}
D_{t_1}^2\tau_{2n}\cdot\tau_{2n}=2\tau_{2n-2}\tau_{2n+2},
\end{align*}
which is indeed a Toda lattice. It is not surprising that the Toda lattice has a Pfaffian tau function since there is a one-to-one correspondence between the Toeplitz-type Pfaffian and the Hankel determinant \cite[Prop. 2.3]{stembridge90}
\begin{align*}
\Pf(\mu_{j-i})_{i,j=1}^{2n}=\det\left(
x_{i,j}
\right)_{i,j=1}^n, \quad x_{i,j}=\mu_{|i-j|+1}+\cdots+\mu_{i+j+1},
\end{align*}
and thus the Hankel determinant solution of the Toda lattice has a Pfaffian version when the evolution is properly chosen. In recent paper \cite{miki20}, the author showed that how to write down a Hankel determinant in terms of Toeplitz-type Pfaffian.

Moreover, it is of interest to obtain the Lax pair of Toda lattice in terms of SOPs. For this purpose, we study the evolution of the eigenvectors under the $t_1$-flow. Besides identity \eqref{derivative}, one can establish the following proposition.
\begin{proposition}
The following identity holds
\begin{align}\label{derivative2}
\frac{1}{\tau_{2n}}(z+\p_{t_1})\Big(\tau_{2n}P_{2n+1}(z)\Big)=P_{2n+2}(z)+(\mathcal{A}_{n}+\mathcal{A}_{n+1})P_{2n+1}(z)-\mathcal{D}_nP_{2n}(z)+\mathcal{B}_n P_{2n-2}(z),
\end{align}
where $\mathcal{A}_n$ and $\mathcal{B}_n$ are given in \eqref{coe1} and
\begin{align*}
\mathcal{D}_n=\frac{s_2(-\tilde{\p}_t)\tau_{2n+2}}{\tau_{2n+2}}+\frac{s_2(\tilde{\p}_t)\tau_{2n}}{\tau_{2n}},\quad s_2(t)=t_2+\frac{1}{2}t_1^2.
\end{align*}
\end{proposition}
\begin{proof}
The proof are based on three steps. The first one is to show
\begin{align*}
(z+\p_{t_1})(\tau_{2n}P_{2n+1}(z))=\Pf(0,\cdots,2n-1,2n+2,z)+\Pf(0,\cdots,2n-2,2n,2n+1,z).
\end{align*}
This step is an analogue of Proposition \ref{prop1}, and we omit the details here.
Then, by using the Pfaffian identity
\begin{align*}
\Pf(\ast,2n,&2n+1,2n+2,z)\Pf(\ast)=\Pf(\ast,2n,2n+1)\Pf(\ast,2n+2,z)\\
&-\Pf(\ast,2n,2n+2)\Pf(\ast,2n+1,z)+\Pf(\ast,2n,z)\Pf(\ast,2n+1,2n+2),
\end{align*}
with $\{\ast\}=\{0,\cdots,2n-1\}$, one can obtain that
\begin{align*}
\Pf(0,\cdots,2n-1,2n+2,z)=\tau_{2n}\left(P_{2n+2}(z)+\mathcal{A}_{n+1}P_{2n+1}(z)-\frac{s_2(-\tilde{\p}_t)\tau_{2n+2}}{\tau_{2n+2}}P_{2n}(z)\right).
\end{align*}
Moreover, the Pfaffian identity
\begin{align*}
\Pf(\ast,2n-1)&\Pf(\ast,2n,2n+1,z)=\Pf(\ast,2n)\Pf(\ast,2n-1,2n+1,z)\\
&-\Pf(\ast,2n+1)\Pf(\ast,2n-1,2n,z)+\Pf(\ast,z)\Pf(\ast,2n-1,2n,2n+1)
\end{align*}
with $\{\ast\}=\{0,\cdots,2n-2\}$  leads to
\begin{align*}
\Pf(0,\cdots,2n-2,2n,2n+1,z)=\tau_{2n}\left(
\mathcal{A}_nP_{2n}(z)-\frac{s_2(\tilde{\p}_t)\tau_{2n}}{\tau_{2n}}P_{2n}(z)+\mathcal{B}_n P_{2n-2}(z)
\right).
\end{align*}
Combining these results gives \eqref{derivative2}.
\end{proof}

Therefore, with the help of \eqref{toda1}-\eqref{toda2}, one can get the time evolutions for the Laurent SOPs
\begin{align*}
\p_{t_1}P_{2n}(z)-\mathcal{B}_n\p_{t_1}P_{2n-2}(z)&=-\mathcal{B}_n P_{2n-1}(z)+\mathcal{A}_{n-1}\mathcal{B}_n P_{2n-2}(z),\\
\p_{t_1}P_{2n+1}(z)-\mathcal{A}_{n+1}\p_{t_1}P_{2n}(z)&=\left(
\mathcal{A}_n\mathcal{A}_{n+1}-\mathcal{D}_n+1
\right)P_{2n}(z)+\mathcal{B}_nP_{2n-2}(z).
\end{align*}
The compatibility condition of the spectral problem and time evolutions give rise to the Toda lattice
\begin{align*}
\p_{t_1}\mathcal{B}_n=\mathcal{B}_n(\mathcal{C}_n-\mathcal{C}_{n-1}),\quad \p_{t_1}\mathcal{C}_n=\mathcal{B}_{n+1}-\mathcal{B}_n,
\end{align*}
where $\mathcal{C}_n=\mathcal{A}_{n+1}-\mathcal{A}_n$.
\begin{remark}
It is not surprising that Toda lattice has a Pfaffian tau function with wave vector SOPs. In the earlier work of Kodama and Pierce \cite{kodama07}, the authors showed that after some certain moment constraints, SOPs $\{P_{n}(z)\}_{n\in\mathbb{N}}$ are connected with standard OPs $
\{p_n(z)\}_{n\in\mathbb{N}}$ such that $P_n(z)=p_n(z^2)$, and the Pfaff lattice becomes Toda lattice.
\end{remark}

\section{Applications of PSOPs' Christoffel transformation}\label{psop_ct}
 In this part, we derive integrable hierarchies with regards to the adjacent family of PSOPs.

\subsection{General case}
In the most general case, we consider the commuting flows such that
\begin{align}\label{te2}
\p_{t_n}\mu_{i,j}=\mu_{i+n,j}+\mu_{i,j+n}, \quad \p_{t_n}\beta_j^{(k)}=\beta_{j+n}^{(k)}.
\end{align}
By \cite[Prop. 4.4]{li19}, the coefficients of the multi-component PSOPs can be expressed in terms of Schur polynomials acting on the normalisation factors.
\begin{proposition}
With time evolutions \eqref{te2}, the multi-component PSOPs have the form
\begin{align}\label{psopschur}
Q^{(m)}_{2n}(z)=\frac{1}{\tau_{2n}^{(m)}}\sum_{\ell=0}^{2n}z^{2n-\ell}s_\ell(-\tilde{\p}_t)\tau_{2n}^{(m)},\quad Q^{(m)}_{2n+1,k}(z)=\frac{1}{\tau_{2n+1,k}^{(m)}}\sum_{\ell=0}^{2n+1}z^{2n+1-\ell}s_\ell(-\tilde{\p}_t)\tau_{2n+1,k}^{(m)}.
\end{align}
\end{proposition}
Substituting the series sum into the Christoffel transformations \eqref{psop-ct1}-\eqref{psop-ct2}, one can obtain the following bilinear integrable hierarchy
\begin{align}\label{ih}
\begin{aligned}
&\tau_{2n}^{(m+1)}s_{2n+1-\ell_1}(-\tilde{\p}_t)\tau_{2n+1,k}^{(m)}+\tau_{2n+1,k}^{(m+1)}s_{2n-\ell_1}(-\tilde{\p}_t)\tau_{2n}^{(m)}\\
&\qquad\qquad\qquad\qquad=\tau_{2n+1,k}^{(m)}s_{2n+1-\ell_1}(-\tilde{\p}_t)\tau_{2n}^{(m+1)}+\tau_{2n+2}^{(m)}s_{2n-\ell_1}(-\tilde{\p}_t)\tau_{2n-1,k}^{(m+1)},\\
&\tau_{2n+1,k}^{(m+1)}s_{2n+2-\ell_2}(-\tilde{\p}_t)\tau_{2n+2}^{(m)}+\tau_{2n+2}^{(m+1)}s_{2n+1-\ell_2}(-\tilde{\p}_t)\tau_{2n+1,k}^{(m)}\\
&\qquad\qquad\qquad\qquad=\tau_{2n+2}^{(m)}s_{2n+2-\ell_2}(-\tilde{\p}_t)\tau_{2n+1,k}^{(m+1)}+\tau_{2n+3,k}^{(m)}s_{2n+1-\ell_2}(-\tilde{\p}_t)\tau_{2n}^{(m+1)},
\end{aligned}
\end{align}
where the first nontrivial example is the case of $\ell_1=2n$, $\ell_2=2n+1$
\begin{subequations}
\begin{align}
\tau_{2n+2}^{(m)}\tau_{2n-1,k}^{(m+1)}&=D_{t_1}\tau_{2n}^{(m+1)}\cdot\tau_{2n+1,k}^{(m)}+\tau_{2n+1,k}^{(m+1)}\tau_{2n}^{(m)},\label{glv1}\\
\tau_{2n+3,k}^{(m)}\tau_{2n}^{(m+1)}&=D_{t_1}\tau_{2n+1,k}^{(m+1)}\cdot\tau_{2n+2}^{(m)}+\tau_{2n+2}^{(m+1)}\tau_{2n+1,k}^{(m)}.\label{glv2}
\end{align}
\end{subequations}
This integrable hierarchy involves two different families of Pfaffian tau functions, namely, the even and odd ones, and it was called the large BKP hierarchy in \cite{kac98,vandeleur15}.
The one-component case of this hierarchy leads to the so-called semi-discrete generalised Lotka-Volterra lattice
\begin{align}\label{glv}
\tau_{n+2}^{(m)}\tau_{n-1}^{(m+1)}=D_{t_1}\tau_{n}^{(m+1)}\cdot\tau_{n+1}^{(m)}+\tau_n^{(m)}\tau_{n+1}^{(m+1)}.
\end{align}
Note that the one-component PSOPs has been considered in \cite{chang182} with $t_1$-flow involved.
Moreover, with the help of adjacent family of PSOPs, one can give a Lax-type pair for the integrable lattice. For simplicity,
we consider the one-component case. However, the result can be easily extended to the multi-component case.
\begin{proposition}
The time-dependent PSOPs have the following properties
\begin{align}\label{mixed}
(z+\p_{t_1})Q_{n}^{(m)}(z)&=Q_{n+1}^{(m)}(z)+\mathcal{K}_n^mQ_{n}^{(m)}(z)-\mathcal{J}_n^mQ_{n-1}^{(m)}(z),
\end{align}
where
\begin{align}\label{jk}
\mathcal{K}_n^m=\p_{t_1}\log\frac{\tau_{n+1}^{(m)}}{\tau_n^{(m)}},\quad \mathcal{J}_n^m=\frac{\tau_{n+2}^{(m)}\tau_{n-1}^{(m)}}{\tau_{n}^{(m)}\tau_{n+1}^{(m)}}.
\end{align}
\end{proposition}
\begin{proof}
The odd and even cases are almost the same, so we consider the odd case.
First, as in the proof of Prop. \ref{prop1}, we have
\begin{align}\label{eq1}
\frac{1}{\tau_{2n+1}^{(m)}}(z+\p_{t_1})(\tau_{2n+1}^{(m)}Q_{2n+1}^{(m)}(z))=z^{-m}\Pf(d,m,\cdots,m+2n,m+2n+2,z).
\end{align}
Choosing $\{\ast\}=\{m,\cdots,m+2n\}$ and using Pfaffian identity
\begin{align*}
\Pf(d,\ast)&\Pf(\ast,m+2n+1,m+2n+2,z)=\Pf(\ast,m+2n+1)\Pf(d,\ast,m+2n+2,z)\\
&-\Pf(\ast,m+2n+2)\Pf(d,\ast,m+2n+1,z)+\Pf(\ast,z)\Pf(d,\ast,m+2n+1,m+2n+2),
\end{align*}
one finds that the right-hand side of equation \eqref{eq1} is indeed
\begin{align*}
Q_{2n+2}^{(m)}(z)+\p_{t_1}\log\tau_{2n+2}^{(m)} Q_{2n+1}^{(m)}(z)-\mathcal{J}_{2n+1}^m Q_{2n}^{(m)}(z).
\end{align*}
Substituting it into \eqref{eq1} provides \eqref{mixed}-\eqref{jk} for odd $n$.
\end{proof}
Therefore, if we denote $\Phi^{(m)}=\left(
Q_0^{(m)}(z),Q_1^{(m)}(z),\cdots
\right)^\top$, then the Christoffel transformation of PSOPs \eqref{psop-ct}  can be rewrite as follows
\begin{align}\label{lax11}
z\Phi^{(m+1)}=L^{(m)}\Phi^{(m)},\quad L^{(m)}=(\mathbf{I}+\Lambda\eta^{(m)})^{-1}(\Lambda^\top+\xi^{(m)}\mathbf{I} )
\end{align}
where $\Lambda$ is the shift operator, $\eta^{(m)}=\text{diag}(\eta_1^m,\eta_2^m,\cdots)$, $\xi^{(m)}=\text{diag}(\xi_0^m,\xi_1^m,\cdots)$. The equation \eqref{mixed} can now be rewritten as
\begin{align}\label{lax21}
(z+\p_{t_1})\Phi^{(m)}=M^{(m)}\Phi^{(m)},\quad M^{(m)}=\Lambda^\top+K^{(m)}\mathbf{I}-\Lambda J^{(m)}
\end{align}
with $K^{(m)}=\text{diag}(\mathcal{K}_0^m,\mathcal{K}_1^m,\cdots)$ and $J^{(m)}=\text{diag}(\mathcal{J}_0^m,\mathcal{J}_1^m,\cdots)$. The compatibility condition of \eqref{lax11} and \eqref{lax21} gives us
\begin{align*}
\p_{t_1}L^{(m)}=M^{(m+1)}L^{(m)}-L^{(m)}M^{(m)}.
\end{align*}

Moreover, for the discrete $t_1$-flow,
\begin{align*}
\mu_{i,j}^{t+1}=\mu_{i+1,j+1}^t+\lambda\mu_{i,j+1}^t+\lambda\mu_{i+1,j}^t+\lambda^2\mu_{i,j}^t,\quad \beta^{t+1}_j=\beta_{j+1}^t+\lambda\beta_j^t,
\end{align*}
the corresponding discrete PSOPs, adjacent family of PSOPs, and integrable lattice were given in \cite{chang182} as well.
In the following, some reductional cases are emphasised.

Inspired by the fact that the odd-order tau functions are independent on the even ones, namely,
for each $k\in\{1,\cdots,\ell\}$, $\{\tau_{2n+1,k}\}$ in \eqref{glv1}-\eqref{glv2} solve \eqref{glv}.
Moreover, since the single moments $\{\beta_j^{k}\}_{j\in\mathbb{N}}$ are independent with bi-moments $\{\mu_{i,j}\}_{i,j\in\mathbb{N}}$, we can expect that there are some relations between these moments.  Such relations are called as {\it moment constraints}. Using these constraints we can impose reductions on PSOPs and corresponding integrable lattices.
We start with the one-component case which can be easily extended to the multi-component case.

\subsection{Moment constraint \Rmnum{1}}
Consider the Laurent type PSOPs satisfying the relations
\begin{align*}
\mu_{i,j}=\mu_{i-1,j-1},\quad \beta_j=\beta_{j-1}.
\end{align*}
Similar to the Laurent SOPs case, one can check that $\tau_{2n}^{(m)}=\tau_{2n}^{(m+1)}$ and $\tau_{2n+1}^{(m)}=\tau_{2n+1}^{(m+1)}$. Also notice that $\{\beta_j\}_{j\in\mathbb{N}}$ are the same functions dependent on $t$, which means the multi-component case no longer exists. Therefore, from the expressions \eqref{psopschur}, one knows that the Christoffel transformation \eqref{psop-ct} is just a three term recurrence relation
\begin{align*}
Q_{n+1}(z)+ Q_{n}(z)=z(Q_n(z)+\eta_nQ_{n-1}(z)),\quad \eta_n=\frac{\tau_{n-1}\tau_{n+2}}{\tau_n\tau_{n+1}}.
\end{align*}
Moreover, if we assume the evolution
\begin{align*}
\p_{t_n}\mu_{i,j}=\mu_{i+n,j}+\mu_{i,j+n},\quad \p_{t_n}\beta_j=\beta_j,
\end{align*}
holds, then \eqref{mixed} is independent on the index $m$.  So one has
\begin{align*}
(z+\p_{t_1})Q_n(z)=Q_{n+1}(z)+\xi_n Q_n(z)-\eta_n Q_{n-1}(z),\quad \xi_n=\p_{t_1}\log \frac{\tau_{n+1}}{\tau_n}.
\end{align*}
 From the above identities, one can derive the time evolutions for the Laurent type PSOPs. Thus we have
 \begin{align*}
 \p_{t_1}Q_{n}+\eta_n\p_{t_1}Q_{n-1}=(\xi_n+\eta_n-1)Q_n+\eta_n(\xi_{n-1}-1)Q_{n-1}(z)-\eta_{n-1}\eta_n Q_{n-2}(z).
 \end{align*}
Since $\{Q_n(z)\}_{n\in\mathbb{N}}$ are monic PSOPs,  $\xi_n+\eta_n-1$  is equal to $0$. This fact is due to the reduction of the integrable hierarchy \eqref{ih}. Under the reduction, one has
 \begin{align}\label{1dlv}
 \tau_ns_{n+1-\ell}(-\tilde{\p}_t)\tau_{n+1}+\tau_{n+1}s_{n-\ell}(-\tilde{\p}_t)\tau_n
 =\tau_{n+1}s_{n+1-\ell}(-\tilde{\p}_t)\tau_n+\tau_{n+2}+s_{2n-\ell}(-\tilde{\p}_t)\tau_{n-1}.
 \end{align}
The famous Lotka-Volterra lattice
 \begin{align*}
 \tau_{n-1}\tau_{n+2}=(D_{t_1}+1)\tau_n\cdot\tau_{n+1}
 \end{align*}
 is the first example of \eqref{1dlv}.
 Therefore, a reduction of the Lax pair \eqref{lax11} and \eqref{lax21} can be obtained. That is, if we denote $\Phi=\left(Q_0(z),Q_1(z),\cdots\right)^\top$,  then
 \begin{align*}
 z\Phi=(\mathbf{I}+\Lambda\eta)^{-1}(\mathbf{I}+\Lambda^\top)\Phi,\quad \p_{t_1}\Phi=(\eta+\Lambda^\top)^{-1}(\mathfrak{a}+\mathfrak{b}\Lambda)\Phi.
 \end{align*}
 where $\eta=\text{diag}(\eta_1,\eta_2,\cdots)$, $\mathfrak{a}=\text{diag}(\eta_1(\xi_0-1),\eta(\xi_1-1),\cdots)$ and $\mathfrak{b}=\text{diag}(\eta_1\eta_2,\eta_2\eta_3,\cdots)$.

Note that the two families of tau functions
 \begin{align*}
 \tau_{2n}=\Pf(0,1,\cdots,2n-1),\quad \tau_{2n+1}=\Pf(d,0,\cdots,2n)
 \end{align*}
 with $\Pf(i,j)=\mu_{i,j}$ and $\Pf(d,i)=\beta$ are two different solutions of the 1d-Toda hierarchy, and thus the integrable hierarchy \eqref{1dlv} acts as the B\"acklund transformation of the 1d-Toda hierarchy \eqref{1dtoda}, which is a classical result in soliton theory.

\subsection{Moment constraint \Rmnum{2}}

The case we consider here is a rank two shift condition
\begin{align}\label{mc2}
\mu_{i,j+1}+\mu_{i+1,j}=\beta_{i+1}\beta_j-\beta_i\beta_{j+1}\quad\text{or \quad $\mathcal{U}\Lambda+\Lambda^\top\mathcal{U}=\Lambda^\top\beta\beta^\top-\beta\beta^\top\Lambda$},
\end{align}
where  $\mathcal{U}=(u_{i,j})_{i,j\in\mathbb{N}}$, $\Lambda$ is the shift operator, and $\beta=(\beta_0,\beta_1,\cdots)^\top$.
One can check that under this assumption, the evolutions $\p_{t_n}\beta_i=\beta_{i+n}$ and $\p_{t_n}\mu_{i,j}=\mu_{i+n,j}+\mu_{i,j+n}$ are consistent.
If we consider only one family of PSOPs, namely, for fixed $m$ we consider $\{Q^{(m)}_n(z)\}_{n\in\mathbb{Z}}$, then it is related to the so-called B-Toda lattice. The details of the corresponding integrable lattice and Lax pair was discussed in \cite{chang182}, and the corresponding integrable hierarchy was given in \cite{li19}. For completeness, we give a brief review here.

From the moment constraint \eqref{mc2}, one can find that the derivative of $t_1$-flow has two different expressions. One expression is the commuting flow, and the other one is in terms of single moments. If two new labels $d_0$ and $d_1$ are introduced such that
\begin{align*}
\p_{t_1}\Pf(i,j)=\Pf(i+1,j)+\Pf(i,j+1)=\Pf(d_0,d_1,i,j),\quad \p_{t_1}\Pf(d_0,i)=\Pf(d_0,i+1)=\Pf(d_1,i),
\end{align*}
then one gets
\begin{align}\label{derpsop}
\begin{aligned}
\p_{t_1}(\tau_{2n}^{(m)}Q_{2n}^{(m)}(z))&=z^{-m}\tau_{2n}^{(m)}\Pf(d_0,d_1,m,\cdots,m+2n,z),\\
\p_{t_1}(\tau_{2n+1}^{(m)}Q_{2n+1}^{(m)}(z))&=z^{-m}\tau_{2n+1}^{(m)}\Pf(d_1,m,\cdots,m+2n+1,z).
\end{aligned}
\end{align}
Using the Pfaffian identities
\begin{align*}
\Pf(d_0,d_1,\ast,m+2n,z)&\Pf(\ast)=\Pf(d_0,d_1,\ast)\Pf(\ast,m+2n,z)\\
&-\Pf(d_0,\ast,m+2n)\Pf(d_1,\ast,z)+\Pf(d_0,\ast,z)\Pf(d_1,\ast,m+2n)
\end{align*}
with $\{\ast\}=\{m,\cdots,m+2n-1\}$ and
\begin{align*}
\Pf(d_0,d_1,\ast,m+2n+1)&\Pf(\ast,z)=\Pf(d_0,\ast)\Pf(d_1,\ast,m+2n+1,z)\\
&-\Pf(d_1,\ast)\Pf(d_0,\ast,m+2n+1,z)+\Pf(\ast,m+2n+1)\Pf(d_0,d_1,\ast,z)
\end{align*}
with $\{\ast\}=\{m,\cdots,m+2n\}$, one can find the following time evolutions
\begin{align}\label{evo}
\p_{t_1}Q_n^{(m)}(z)+\mathcal{I}_n^m \p_{t_1}Q_{n-1}^{(m)}(z)=\mathcal{I}_n^m(\mathcal{K}_n^m+\mathcal{K}_{n-1}^m)Q_{n-1}^{(m)}(z),\quad \mathcal{I}_n^m=\frac{\tau_{n+1}^{(m)}\tau_{n-1}^{(m)}}{(\tau_n^{(m)})^2},
\end{align}
where $\mathcal{K}_n^m$ is defined in \eqref{jk}.
The compatibility condition of \eqref{mixed} and \eqref{evo} gives us the B-Toda lattice
\begin{align}\label{btoda}
D_{t_1}^2 \tau_n^{(m)}\cdot\tau_n^{(m)}=2D_{t_1}\tau_{n+1}^{(m)}\cdot\tau_{n-1}^{(m)}.
\end{align}
Since we focus on the adjacent PSOPs in this paper, in what follows, we show how to apply the Christoffel transformation to this kind of moment constraint.

By using the derivative formula \eqref{derpsop} and Pfaffian identities
\begin{align*}
\Pf(d_0,d_1,m,\ast,z)&\Pf(\ast)=\Pf(d_0,d_1,\ast)\Pf(m,\ast,z)\\
&-\Pf(d_0,m,\ast)\Pf(d_1,\ast,z)+\Pf(d_0,\ast,z)\Pf(d_1,m,\ast)
\end{align*}
with $\{\ast\}=\{m+1,\cdots,m+2n\}$ and
\begin{align*}
\Pf(d_0,d_1,m,\ast)&\Pf(\ast,z)=\Pf(d_0,\ast)\Pf(d_1,m,\ast,z)\\
&-\Pf(d_1,\ast)\Pf(d_0,m,\ast,z)+\Pf(d_0,d_1,\ast,z)\Pf(m,\ast)
\end{align*}
with $\{\ast\}=\{m+1,\cdots,m+2n+1\}$, we find the following formula
\begin{align}\label{psoplax1}
\begin{aligned}
\p_{t_1}Q_n^{(m)}(z)&+\p_{t_1}\log\frac{\tau_{n}^{(m)}}{\tau_n^{(m+1)}}Q_n^{(m)}(z)\\
&=z\left(
\frac{\tau_{n+1}^{(m)}\tau_{n-1}^{(m+1)}}{\tau_n^{(m)}\tau_n^{(m+1)}}\p_{t_1}Q_{n-1}^{(m+1)}(z)+\frac{D_{t_1}\tau_{n+1}^{(m)}\cdot\tau_{n-1}^{(m+1)}}{\tau_n^{(m)}\tau_{n}^{(m+1)}}Q_{n-1}^{(m+1)}(z)
\right),
\end{aligned}
\end{align}
which is of degree $n$ on both sides, and therefore, it must be
\begin{align*}
D_{t_1}\tau_n^{(m)}\cdot\tau_n^{(m+1)}=D_{t_1}\tau_{n+1}^{(m)}\cdot\tau_{n-1}^{(m+1)}.
\end{align*}
This integrable lattice is the B\"acklund transformation of the B-Toda lattice \eqref{btoda}.
We now proceed to the Lax-type equation. Noting that the equation \eqref{psoplax1} can be rewritten as
\begin{align*}
Q_n^{(m)}(z)+A_n^m \p_{t_1}Q_n^{(m)}(z)=z(Q_{n-1}^{(m+1)}(z)+B_n^m Q_{n-1}^{(m+1)}(z)),
\end{align*}
where the coefficients
\begin{align*}
A_n^m=(\p_{t_1}\log\frac{\tau_n^{(m)}}{\tau_n^{(m+1)}})^{-1},\quad B_n^m=\frac{\tau_{n+1}^{(m)}\tau_{n-1}^{(m+1)}}{D_{t_1}\tau_{n+1}^{(m)}\cdot\tau_{n-1}^{(m+1)}}.
\end{align*}
Combining the Christoffel transformation \eqref{psop-ct} and the formulas above, one can find
\begin{align}\label{b1}
z\p_{t_1}\Phi^{(m+1)}=L_1^{(m)}\p_{t_1}\Phi^{(m)}+L_2^{(m)}\Phi^{(m)},\quad \Phi^{(m)}=\left(
Q_0^{(m)}(z),Q_1^{(m)}(z),\cdots
\right)^\top,
\end{align}
where the coefficients matrices
\begin{align*}
L_1^{(m)}=(\mathfrak{a}_1+\Lambda\mathfrak{a}_2)^{-1}(\mathfrak{a}_3+\mathfrak{a}_4\Lambda^\top),\quad L_2^{(m)}=(\mathfrak{a}_1+\Lambda\mathfrak{a}_2)^{-1}\mathfrak{a}_5
\end{align*}
with $\mathfrak{a}_1=\text{diag}(B_1^m,B_2^m,\cdots),\,\mathfrak{a}_2=\text{diag}(\eta_1^mB_1^m,\eta_2^mB_2^m,\cdots),\,\mathfrak{a}_3=\text{diag}(\eta_0^mA_0^m,\eta_1^mA_1^m,\cdots)$, $\mathfrak{a}_4=\text{diag}(A_1^m,A_2^m,\cdots)$
and $\mathfrak{a}_5=\text{diag}(\eta_0^m-\xi_0^m,\eta_1^m-\xi_1^m,\cdots)$.
Moreover, with the notation $\mathfrak{b}_1=\text{diag}(\mathcal{I}_1^m,\mathcal{I}_2^m,\cdots),\,\mathfrak{b}_2=\text{diag}(\mathcal{I}_0^m(\mathcal{K}_0^m+\mathcal{K}_1^m),\cdots)$, the time evolution \eqref{evo} can be rewritten as
\begin{align}\label{b2}
\p_{t_1}\Phi^{(m)}=M^{(m)}\Phi^{(m)}, \quad M^{(m)}=(\mathfrak{b}_1+\Lambda^\top)^{-1}\mathfrak{b}_2.
\end{align}
and by denoting $\mathfrak{b}_3=\text{diag}(\xi_0^m,\xi_1^m,\cdots)$ and $\mathfrak{b}_4=\text{diag}(\eta_1^m,\eta_2^m,\cdots)$, so that the Christoffel transformation \eqref{psop-ct} has the form
\begin{align}\label{b3}
\Phi^{(m)}=zN^{(m)}\Phi^{(m+1)},\quad N^{(m)}=(\mathfrak{b}_3+\Lambda^\top)^{-1}(\mathbf{I}+\Lambda\mathfrak{b}_4).
\end{align}
The compatibility conditions of \eqref{b1}, \eqref{b2} and \eqref{b3} give two alternative forms
\begin{align*}
M^{(m+1)}=(L_1^{(m)}M^{(m)}+L_2^{(m)})N^{(m)},\quad \p_{t_1}N^{(m)}=\left(M^{(m)}-N^{(m)}(L_1^{(m)}M^{(m)}+L_2^{(m)})\right)N^{(m)}.
\end{align*}

\subsection{Moment constraint \Rmnum{3}}
In this section, we consider the moment constraint
\begin{align}\label{cons1}
\mu_{i,j+1}-\mu_{i+1,j}=2\beta_i\beta_j \quad \text{or \quad $\mathcal{U}\Lambda-\Lambda^\top \mathcal{U}=2\beta\beta^\top$},
\end{align}
where $\mathcal{U}=(u_{i,j})_{i,j\in\mathbb{N}}$, $\Lambda$ is the shift operator, and $\beta=(\beta_0,\beta_1,\cdots)^\top$.
From this constraint, we can express the bi-moments $\{\mu_{i,j}\}_{i,j\in\mathbb{N}}$ in terms of single moments $\{\beta_i\}_{i\in\mathbb{N}}$ as follows
\begin{align*}
\mu_{i,i+2k+1}=2\sum_{s=0}^{k-1}\beta_{i+s}\beta_{i+2k-s}+\beta_{i+k}^2,\quad \mu_{i,i+2k}=2\sum_{s=0}^{k-1}\beta_{i+s}\beta_{i+2k-1-s}.
\end{align*}
This kind of constraint is amount to
\begin{align}\label{evod}
\tau_{2n}^{(m)}\tau_{2n+2}^{(m)}=(\tau_{2n+1}^{(m)})^2,
\end{align}
followed by the observation
\begin{align*}
\Pf(d,m,\cdots,2n+m)\Pf(d,m,\cdots,2n+m)=\Pf(m,\cdots,2n-1+m)\Pf(m,\cdots, 2n+1+m).
\end{align*}

Moreover, the time evolutions on $\{\mu_{i,j}\}_{i,j\in\mathbb{N}}$ is dependent on $\{\beta_i\}_{i\in\mathbb{N}}$, so we can state the following proposition.
\begin{proposition}\label{prop43}
If the time evolutions on $\{\beta_i\}_{i\in\mathbb{N}}$ satisfy
\begin{align*}
\p_{t_{n}}\beta_i=\beta_{i+n},
\end{align*}
then the evolutions of $\{\mu_{i,j}\}_{i,j\in\mathbb{N}}$ satisfy
\begin{align*}
\p_{t_{n}}\mu_{i,j}=\mu_{i,j+n}+\mu_{i+n,j}.
\end{align*}
\end{proposition}
\begin{proof}
Here we prove the time evolution on $\mu_{i,i+2k}$, while for $\mu_{i,i+2k+1}$ the statement of time evolution could be verified similarly. For odd-order time flow $t_{2n+1}$, we have
\begin{align*}
\p_{t_{2n+1}}\mu_{i,i+2k}=2\sum_{s=0}^{k-1}\left(\beta_{i+2n+1+s}\beta_{i+2k-1-s}+
\beta_{i+s}\beta_{i+2k+2n-s}
\right)
\end{align*}
and
\begin{align*}
&\sum_{s=0}^{k-1}\beta_{i+s}\beta_{i+2(k+n)-s}=\sum_{s=0}^{n+k-1}\beta_{i+s}\beta_{i+2(k+n)-s}-\sum_{s=0}^{n-1}\beta_{i+k+s}\beta_{i+2n+k-s},\\
&\sum_{s=0}^{k-1}\beta_{i+2n+1+s}\beta_{i+2k-1-s}=\sum_{s=n+1}^{k-1}\beta_{i+n+s}\beta_{i+2k+n-s}+\sum_{s=0}^n\beta_{i+k+s}\beta_{i+2n+k-s},
\end{align*}
So one can immediately get the result.
For the even-order flow $t_{2n}$, we have
\begin{align*}
\p_{t_{2n}}\mu_{i,i+2k}=2\sum_{s=0}^{k-1}(
\beta_{i+2n+s}\beta_{i+2k-s-1}+\beta_{i+s}\beta_{i+2k+2n-1-s}
).
\end{align*}
Since the following identities hold
\begin{align*}
\sum_{s=0}^{k-1}\beta_{i+2n+s}\beta_{i+2k-s-1}=\sum_{s=n}^{k-1}\beta_{i+n+s}\beta_{i+2k+n-s-1}+\sum_{s=0}^{n-1}\beta_{i+m+s}\beta_{i+2n+m-s-1},\\
\sum_{s=0}^{k-1}\beta_{i+s}\beta_{i+2k+2n-s-1}=\sum_{s=0}^{k+n-1}\beta_{i+s}\beta_{i+2n+2k-s-1}-\sum_{s=0}^{n-1}\beta_{i+k+s}\beta_{i+2n+k-s-1},
\end{align*}
the derivatives for even flows could be verified.
\end{proof}
Therefore, the evolutions on the moments satisfy the equation \eqref{te2}, and we can get the one-component integrable system
\begin{align*}
&\tau_{2n+2}^{(m)}\tau_{2n-1}^{(m+1)}=D_{t_1}\tau_{2n}^{(m+1)}\cdot\tau_{2n+1}^{(m)}+\tau_{2n}^{(m)}\tau_{2n+1}^{(m+1)},\\
&\tau_{2n+1}^{(m)}\tau_{2n-2}^{(m+1)}=D_{t_1}\tau_{2n-1}^{(m+1)}\cdot\tau_{2n}^{(m)}+\tau_{2n-1}^{(m)}\tau_{2n}^{(m+1)},\\
&\tau_{2n}^{(m)}\tau_{2n+2}^{(m)}=(\tau_{2n+1}^{(m)})^2,\quad \tau_{2n}^{(m+1)}\tau_{2n+2}^{(m+1)}=(\tau_{2n+1}^{(m+1)})^2.
\end{align*}
If we rewrite
\begin{align*}
\tau_{2n}^{(m)}=f_{2n},\quad \tau_{2n}^{(m+1)}=f_{2n+1},\quad \tau_{2n+1}^{(m)}=g_{2n+1},\quad \tau_{2n+1}^{(m+1)}=g_{2n+2},
\end{align*}
then the equation above can be written in a unified form
\begin{align*}
D_{t_1}g_n\cdot f_n-g_{n+1}f_{n-1}+g_{n-1}f_{n+1}=0,\quad f_{n+1}f_{n-1}=g_n^2,
\end{align*}
which is the so-called modified KdV equation \cite{hirota97}.

\begin{proposition}\label{prop2}
The polynomials $\{Q_{2n}^{(m)}(z)\}_{n\in\mathbb{N}}$ satisfy the following three term recurrence relation
\begin{align}\label{lax1}
zQ_{2n}^{(m)}(z)=Q_{2n+1}^{(m)}(z)+\mathcal{K}_{2n}^{m}Q_{2n}^{(m)}(z)+ \mathcal{J}_{2n}^mQ_{2n-1}^{(m)}(z),
\end{align}
where $\mathcal{K}_{2n}^m$ and $\mathcal{J}_{2n}^m$ are given in \eqref{jk}.
\end{proposition}
\begin{proof}
The proof is based on the expansion of the polynomial $Q_{2n+1}^{(m)}$, that is, the expansion of
\begin{align*}
\Pf(d,m,\cdots,m+2n+1,z)\Pf(d,m,\cdots,m+2n).
\end{align*}
Similar to  \eqref{evod}, one can get
\begin{align*}
Q_{2n+1}^{(m)}(z)=\frac{1}{2}zQ_{2n}^{(m)}(z)+\frac{1}{2}\frac{1}{\tau_{2n}^{(m)}z^m}\Pf(m,\cdots,m+2n-1,m+2n+1,z)-\frac{1}{2}\mathcal{I}_{2n+2}^mQ_{2n}^{(m)}(z).
\end{align*}
Now we need to deal with the mid term on the right hand side.
By using the  Pfaffian identity
\begin{align*}
\Pf(d,\ast,&m+2n,m+2n+1,z)\Pf(\ast)=\Pf(d,\ast,m+2n)\Pf(\ast,m+2n+1,z)\\
&\quad-\Pf(d,\ast,m+2n+1)\Pf(\ast,m+2n,z)+\Pf(d,\ast,z)\Pf(\ast,m+2n,m+2n+1)
\end{align*}
with $\{\ast\}=\{m,\cdots,m+2n-1\}$ and equation \eqref{evod}, one could find that it is equal to
\begin{align*}
Q_{2n+1}^{(m)}(z)+\mathcal{I}_{2n+1}^{m}Q_{2n}^{(m)}(z)-\mathcal{J}_n^m Q_{2n-1}^{(m)}(z).
\end{align*}
Noting $\mathcal{K}_{2n}^m=\mathcal{K}_{2n+1}^m$ and combining these gives the result.
\end{proof}

Moreover, the even-order polynomials satisfy the following time evolutions.
\begin{corollary}\label{prop3}
The following time evolutions for the specific $\{Q_{2n}^{(m)}(z)\}_{n\in\mathbb{N}}$ hold
\begin{align}\label{lax2}
\p_{t_1}Q_{2n}^{(m)}(z)=-2\mathcal{J}_{2n}^m Q_{2n-1}^{(m)}(z).
\end{align}
\end{corollary}
\begin{proof}
As was shown in Proposition \ref{prop1}, the following identity is true,
\begin{align*}
(z+\p_{t_1})(\tau_{2n}^{(m)}Q_{2n}^{(m)}(z))=\Pf(m,\cdots,m+2n-1,m+2n+1,z).
\end{align*}
According to the above Pfaffian identity, one has that the right-hand side is
\begin{align*}
\tau_{2n}^{(m)}\left(
Q_{2n+1}^{(m)}(z)+\p_{t_1}\tau_{2n+1}^{(m)} Q_{2n}^{(m)}(z)-\mathcal{J}_{2n}^m Q_{2n-1}^{(m)}(z)
\right).
\end{align*}
By dividing $\tau_{2n}^{(m)}$ on both sides and using \eqref{evod}, we get \eqref{lax2}.
\end{proof}
The spectral problem for the odd-order polynomials is more difficult. The observation
\begin{align}\label{oddpsops}
\langle z^{m+1}Q_{2n+1}^{(m)}(z), z^{m+j}\rangle
=\frac{\Pf(d,m,\cdots,m+2n+1,m+j+1)}{\tau_{2n+1}^{(m)}}-\Pf(d,m+j+1)\frac{\tau_{2n+2}^{(m)}}{\tau_{2n+1}^{(m)}},
\end{align}
implies the following proposition.
\begin{proposition}\label{prop4}
The following spectral problem for the specific odd order PSOPs holds
\begin{align}\label{oddspectral}
z\left(
Q_{2n+1}^{(m)}(z)-\mathcal{J}_{2n}^m Q_{2n-1}^{(m)}(z)
\right)&=Q_{2n+2}^{(m)}+\mathcal{K}_{2n}^m Q_{2n+1}^{(m)}(z)+\left(\alpha_n^m+\mathcal{K}_{2n}^m \p_{t_1}\log\tau_{2n+1}^{(m)}\right)Q_{2n}^{(m)}(z)\nonumber\\
&\quad-\mathcal{K}_{2n}^m\mathcal{J}_{2n}^m Q_{2n-1}^{(m)}(z)-\mathcal{J}_{2n-1}^m \mathcal{J}_{2n}^mQ_{2n-2}^{(m)}(z),
\end{align}
where $\alpha_n^m$ is given in \eqref{coe1} and $\mathcal{J}_{n}^m$ and $\mathcal{K}_n^m$ are given in \eqref{jk}.
\end{proposition}
\begin{proof}
From \eqref{oddpsops} we obtain
\begin{align*}
&\left\langle z^{m+1}\left(Q_{2n+1}^{(m)}(z)-\mathcal{J}_{2n}^m Q_{2n-1}^{(m)}(z)
\right),z^{m+j}\right\rangle\\
&\quad =
\frac{\Pf(d,m,\cdots,m+2n+1,m+j+1)}{\tau_{2n+1}^{(m)}}-\frac{\tau_{2n+2}^{(m)}\Pf(d,m,\cdots,m+2n-1,m+j+1)}{\tau_{2n}^{(m)}\tau_{2n+1}^{(m)}},
\end{align*}
which is equal to $0$ when $j=0,\cdots,2n-2$. At this step, we need to consider a set of basis for the SOPs  and then transform the set into the basis for the PSOPs. Expanding the left hand of \eqref{oddspectral} in terms of SOPs, we get
\begin{align}\label{expand}
\begin{aligned}
& z\left(Q_{2n+1}^{(m)}(z)-\mathcal{J}_{2n}^m Q_{2n-1}^{(m)}(z)
\right)=Q_{2n+2}^{(m)}(z)\\
&\quad+z^{-m}\left(\sum_{i=0}^{n}\frac{\alpha_i^m}{\tau_{2i}^{(m)}}\Pf(m,\cdots,m+2i,z)+\sum_{i=0}^{n}\frac{\beta_i^m}{\tau_{2i}^{(m)}}\Pf(m,\cdots,m+2i-1,m+2i+1,z)\right).
\end{aligned}
\end{align}
Taking the skew inner product with $\langle z^m\cdot, z^{m+j}\rangle$ for $j=0,\cdots,2n-2$, we obtain that $\alpha_i^m=0$ for $i=0,\cdots,n-2$ and $\beta_j^m=0$ whenever $j=0,\cdots,n-1$. For $j=2n-1$, $j=2n$ and $j=2n+1$, we have
\begin{align}\label{coe1}
\alpha_{n-1}^{m}=-\mathcal{J}_{2n-1}^m\mathcal{J}_{2n}^m,\quad\beta_n^{m}=\mathcal{K}_{2n}^m,\quad \alpha_{n}^m=\mathcal{J}_{2n+1}^{(m)}-\frac{s_2(-\tilde{\p}_t)\tau_{2n+1}^{(m)}}{\tau_{2n+1}^{(m)}}+\frac{s_2(-\tilde{\p}_t)\tau_{2n}^{(m)}}{\tau_{2n}^{(m)}},
\end{align}
where $s_2(-\tilde{\p}_t)$ is the Schur function mentioned before. From Prop. \ref{prop1}, we get
\begin{align*}
&z\left(Q_{2n+1}^{(m)}(z)-\mathcal{J}_{2n}^m Q_{2n-1}^{(m)}(z)
\right)\\
&\quad=Q_{2n+2}^{(m)}(z)+\alpha_n^m Q_{2n}^{(m)}(z)+\alpha_{n-1}^{(m)}Q_{2n-2}^{(m)}(z)+\frac{\beta_n^m}{\tau_{2n}^{(m)}}(z+\p_{t_1})(\tau_{2n}^{(m)}Q_{2n}^{(m)}(z)),
\end{align*}
Substituting $zQ_{2n}^{(m)}(z)$ and $\p_{t_1}Q_{2n}^{(m)}(z)$ into this identity, we obtain \eqref{oddspectral}.
\end{proof}

Analogously to the even-order polynomials, the odd-order polynomials also have the time evolutions.
\begin{corollary}
The specific odd order polynomials $\{Q_{2n+1}^{(m)}(z)\}_{n\in\mathbb{N}}$ admit the following time evolutions
\begin{align}\label{lax4}
\begin{aligned}
\p_{t_1}Q_{2n+1}^{(m)}(z)-\mathcal{J}_{2n}^m\p_{t_1}Q_{2n-1}^{(m)}(z)&=\left(
\mathcal{J}_{2n+1}^m+\mathcal{J}_{2n}^m-\alpha_n^m-\mathcal{K}_n^m\p_{t_1}\log\tau_{2n+1}^{(m)}
\right)Q_{2n}^{(m)}(z)\\
&\quad+\mathcal{J}_{2n}^m\left(
\mathcal{K}_{2n}^m-\mathcal{K}_{2n-1}^m
\right)Q_{2n-1}^{(m)}(z)-2\mathcal{J}_{2n-1}^m\mathcal{J}_{2n}^m Q_{2n-2}^{(m)}(z).
\end{aligned}
\end{align}
\end{corollary}
\begin{proof}
This proof using the following identity
\begin{align*}
\frac{1}{\tau_{2n+1}^{(m)}}(z+\p_{t_1})(\tau_{2n+1}^{(m)}Q_{2n+1}^{(m)})=z^{-m}\Pf(d,m,\cdots,m+2n,m+2n+2,z).
\end{align*}
From the Pfaffian identity
\begin{align*}
\Pf(d,\ast,&m+2n+1,m+2n+2)\Pf(\ast,z)=\Pf(d,\ast)\Pf(\ast,m+2n+1,m+2n+2,z)\\
&-\Pf(\ast,m+2n+1)\Pf(d,\ast,m+2n+2,z)+\Pf(\ast,m+2n+2)\Pf(d,\ast,m+2n+1,z),
\end{align*}
with $\{\ast\}=\{m,\cdots,m+2n\}$, we obtain
\begin{align*}
(z+\p_{t_1})Q_{2n+1}^{(m)}(z)=Q_{2n+2}^{(m)}(z)+\mathcal{K}_{2n+1}^{(m)}Q_{2n+1}^{(m)}(z)-\mathcal{J}_{2n+1}^{m}Q_{2n}^{(m)}(z).
\end{align*}
Now Proposition \ref{prop4} implies \eqref{lax4}.
\end{proof}
Thus, denoting $\Phi=\left(
Q_0^{(m)}(z),Q_1^{(m)}(z),\cdots
\right)^\top$,
we get the Lax pair
\begin{align*}
z\Phi=L\Phi, \quad \p_{t_1}\Phi=M\Phi,
\end{align*}
where $L$ and $M$ are constructed with  \eqref{lax1}, \eqref{lax2}, \eqref{oddspectral} and \eqref{lax4}. Another possible way to find the Lax pair is due to the idea in \cite{chang16}; one can make the use of the SOPs as the eigenfunctions  and regard the PSOPs as the auxiliary polynomials to obtain the Lax pair.

\begin{remark}
The essence of this reduction is a $1+1$ dimension integrable lattice, and therefore the eigenfunction of this Lax pair involves $\{Q_{n}^{(m)}(z)\}_{n\in\mathbb{N}}$ only, which requires a higher-order time flow $t_2$ to take place of $\tau^{(m+1)}_n$.
\end{remark}

\subsubsection{Multi-component case}
From \cite{hirota97} it follows that the above reduction has a multi-component extension. If we extend the constraint \eqref{cons1} to
\begin{align*}
\mu_{i,j+1}-\mu_{i+1,j}=2\sum_{a,b=1}^N \beta_i^{a}\beta_j^b,
\end{align*}
then we can express the bi-moments $\{\mu_{i,j}\}_{i,j\in\mathbb{N}}$ in terms of single moments $\{\beta_j^a\}_{j\in\mathbb{N}}$ for $a=1,\cdots,N$, and
\begin{align*}
\mu_{i,i+2k+1}=\sum_{a,b=1}^N\left(
2\sum_{s=0}^{k-1}\beta_{i+s}^a\beta_{i+2k-s}^b+\beta_{i+k}^a\beta_{i+k}^b
\right),\quad \mu_{i,i+2k}=\sum_{a,b=1}^N \left(
2\sum_{s=0}^{k-1}\beta_{i+k}^a\beta_{i+2k-1-s}^b
\right).
\end{align*}
Similar to Proposition \ref{prop43}, one can show that if $\p_{t_n}\beta_j^a=\beta_{j+n}^a$ for all $a=1,\cdots,N$, then $\mu_{i,j}$ satisfies $\p_{t_n}\mu_{i,j}=\mu_{i+n,j}+\mu_{i,j+n}$. Moreover, the even and odd tau functions are connected with each other by the formula
\begin{align*}
\tau_{2n}^{(m)}\tau_{2n+2}^{(m)}=\sum_{a,b=1}^N \tau_{2n+1,a}^{(m)}\tau_{2n+1,b}^{(m)}.
\end{align*}
Thus we  get the coupled modified KdV (cmKdV) hierarchy
\begin{align*}
&\tau_{2n}^{(m+1)}s_{2n+1-\ell_1}(-\tilde{\p}_t)\tau_{2n+1,k}^{(m)}+\tau_{2n+1,k}^{(m+1)}s_{2n-\ell_1}(-\tilde{\p}_t)\tau_{2n}^{(m)}\\
&\qquad\qquad\qquad\qquad=\tau_{2n+1,k}^{(m)}s_{2n+1-\ell_1}(-\tilde{\p}_t)\tau_{2n}^{(m+1)}+\tau_{2n+2}^{(m)}s_{2n-\ell_1}(-\tilde{\p}_t)\tau_{2n-1,k}^{(m+1)},\\
&\tau_{2n+1,k}^{(m+1)}s_{2n+2-\ell_2}(-\tilde{\p}_t)\tau_{2n+2}^{(m)}+\tau_{2n+2}^{(m+1)}s_{2n+1-\ell_2}(-\tilde{\p}_t)\tau_{2n+1,k}^{(m)}\\
&\qquad\qquad\qquad\qquad=\tau_{2n+2}^{(m)}s_{2n+2-\ell_2}(-\tilde{\p}_t)\tau_{2n+1,k}^{(m+1)}+\tau_{2n+3,k}^{(m)}s_{2n+1-\ell_2}(-\tilde{\p}_t)\tau_{2n}^{(m+1)},\\
&\tau_{2n}^{(m)}\tau_{2n+2}^{(m)}=\sum_{a,b=1}^N \tau_{2n+1,a}^{(m)}\tau_{2n+1,b}^{(m)},\qquad \tau_{2n}^{(m+1)}\tau_{2n+2}^{(m+1)}=\sum_{a,b=1}^N \tau_{2n+1,a}^{(m+1)}\tau_{2n+1,b}^{(m+1)}.
\end{align*}
The first example of this hierarchy is the cmKdV equation.

\subsubsection{Complex multi-component case}
The multi-component case  admits a complex version as well.
If the single moments $\{\beta_j^a\}_{j\in\mathbb{N},a\in\{1,\cdots,N\}}$ are complex and satisfy the time evolutions
\begin{align*}
\p_{t_n}\beta_j^a=\beta_{j+n}^a,\quad \p_{t_n}\bar{\beta}_j^a=\bar{\beta}_{j+n}^a
\end{align*}
then the bi-moments $\{\mu_{i,j}\}_{i,j\in\mathbb{N}}$ can be expressed in terms of single moments as follows.
\begin{align*}
\mu_{i,j+1}-\mu_{i+1,j}=2\sum_{a,b=1}^N \beta_i^a\bar{\beta}_j^b,
\end{align*}
and satisfy the identity $\p_{t_n}\mu_{i,j}=\mu_{i+n,j}+\mu_{i,j+n}$.
Thus, we have
\begin{align*}
\tau_{2n}^{(m)}\tau_{2n+2}^{(m)}=\sum_{a,b=1}^N \tau_{2n+1,a}^{(m)}\bar{\tau}_{2n+1,b}^{(m)}.
\end{align*}
Such kind of constraint may lead us to the discrete vector NLS hierarchy \cite{maruno08}
\begin{align*}
&\tau_{2n}^{(m+1)}s_{2n+1-\ell_1}(-\tilde{\p}_t)\tilde{\tau}_{2n+1,k}^{(m)}+\tilde{\tau}_{2n+1,k}^{(m+1)}s_{2n-\ell_1}(-\tilde{\p}_t)\tau_{2n}^{(m)}\\
&\qquad\qquad\qquad\qquad=\tilde{\tau}_{2n+1,k}^{(m)}s_{2n+1-\ell_1}(-\tilde{\p}_t)\tau_{2n}^{(m+1)}+\tau_{2n+2}^{(m)}s_{2n-\ell_1}(-\tilde{\p}_t)\tilde{\tau}_{2n-1,k}^{(m+1)},\\
&\tilde{\tau}_{2n+1,k}^{(m+1)}s_{2n+2-\ell_2}(-\tilde{\p}_t)\tau_{2n+2}^{(m)}+\tau_{2n+2}^{(m+1)}s_{2n+1-\ell_2}(-\tilde{\p}_t)\tilde{\tau}_{2n+1,k}^{(m)}\\
&\qquad\qquad\qquad\qquad=\tau_{2n+2}^{(m)}s_{2n+2-\ell_2}(-\tilde{\p}_t)\tilde{\tau}_{2n+1,k}^{(m+1)}+\tilde{\tau}_{2n+3,k}^{(m)}s_{2n+1-\ell_2}(-\tilde{\p}_t)\tau_{2n}^{(m+1)},\\
&\tau_{2n}^{(m)}\tau_{2n+2}^{(m)}=\sum_{a,b=1}^N \tau_{2n+1,a}^{(m)}\bar{\tau}_{2n+1,b}^{(m)},\qquad \tau_{2n}^{(m+1)}\tau_{2n+2}^{(m+1)}=\sum_{a,b=1}^N \tau_{2n+1,a}^{(m+1)}\bar{\tau}_{2n+1,b}^{(m+1)},
\end{align*}
where $\tilde{\tau}$ means that $\tau$ and $\bar{\tau}$ both satisfy those equations.

\section{Concluding remarks}\label{con}

In this article, we consider the Christoffel transformation for SOPs and PSOPs along with their applications in integrable systems. The eigenfunctions of the Lax pair are given in terms of SOPs or PSOPs and integrable hierarchies are expressed in terms of the coefficients of polynomials. The advantage of SOPs lies in the fact that the basis is skew orthogonal and therefore, it's better for us to choose the SOPs as the basis to expand  some polynomials, see for example, equation \eqref{expand}. On the other hand, the advantage of PSOPs is that we can naturally introduce the odd-order tau functions that may shed lights into novel integrable hierarchies as well as iterative algorithms, for example, in the design of Grave-Morris' vector Pad\'e approximation. Thus, both of the polynomials have their own strengths and should be properly chosen while using.

\section*{Acknowledgement}
The authors would like to thank Dr. Hiroshi Miki and Prof. Xing-Biao Hu for helpful discussions and comments.
G. Yu is supported by National Natural Science Foundation of China (Grant no. 11871336).

\small
\bibliographystyle{abbrv}

\def\cydot{\leavevmode\raise.4ex\hbox{.}}
  \def\cydot{\leavevmode\raise.4ex\hbox{.}} \def\cprime{$'$}

\end{document}